\documentclass[conference,twocolumn, letter]{IEEEtran}%

\usepackage{eurosym}
\usepackage{amsmath}
\usepackage{amsthm}
\usepackage{amsfonts}
\usepackage{amssymb}
\usepackage{graphicx}
\usepackage{tikz}
\usepackage{bm}
\usepackage{epstopdf}
\usepackage{booktabs}
\usepackage[caption=false,font=footnotesize]{subfig}
\usepackage{booktabs}%
\providecommand{\U}[1]{\protect\rule{.1in}{.1in}}

\theoremstyle{plain}
\newtheorem{theorem}{Theorem}
\newtheorem{corollary}{Corollary}
\newtheorem{lemma}{Lemma}

\newtheorem{problem}{Problem}

\newtheorem{definition}{Definition}

\newtheorem{remark}{Remark}

\title{Overshooting and $L^1$-Norms \\of a Class of Nyquist Filters}

\author{\IEEEauthorblockN{Gerhard Wunder, Saeed Afrasiabi-Gorgani}
\IEEEauthorblockA{Heisenberg Communications and Information Theory Group\\
Freie Universit\"{a}t Berlin\\
wunder@zedat.fu-berlin.de}}

\begin{document}

\maketitle

\begin{abstract}
To tightly control the signal envelope, estimating the peak regrowth between FFT samples is an important sub-problem in multicarrier communications. While the problem is well-investigated for trigonometric polynomials (i.e. OFDM), the impact of an aperiodic transmit filter is important too and typically neglected in the peak regrowth analysis. In this paper, we provide new bounds on the overshooting between samples for general multicarrier signals improving on available bounds for small oversampling factors. In particular, we generalize a result of \cite[Theorem~4.10]{litsyn_07}. Our results will be extended to bound overshooting of a class of Nyquist filters as well. Eventually, results are related to some respective $L^1$-properties of these filters with application to filter design.
\end{abstract}

\begin{IEEEkeywords}
PAPR, peak value, peak regrowth, Nyquist filter, $L_p$-norms, noise enhancement
\end{IEEEkeywords}

\section{Introduction}

The high peak-to-average power ratio (PAPR) of OFDM (or more general:
multicarrier) transmit signals is a crucial problem. It has re-attracted much
attention recently due to 5G \cite{ComMag5GNOW,Wunder2015_ACCESS}.
Multicarrier transmit signals are efficiently generated using some
(oversampled) FFT processing. However, to tightly control the signal envelope,
estimating the overshooting between FFT samples is an important sub-problem
which has been well-investigated in the literature
\cite{Wunder2013_SPM,Wunder2015_ACCESS}, particularly for trigonometric
polynomials. On the other hand, the impact of the aperiodic transmit filter is
important too and typically neglected in the peak regrowth analysis
(representing $\emph{not}$ trigonometric polynomials). Another problem is that
available bounds are quite loose for small oversampling factors which will be
typical in upcoming 5G multicarrier systems \cite{ComMag5GNOW} e.g. in FBMC
transmission using very large FFTs.

In this paper we provide new bounds on the overshooting between samples for
general multicarrier signals improving on available bounds for small
oversampling factors. Our results will be extended to bound overshooting of a
class of Nyquist filters as well. Eventually, results are related to some
respective $\mathcal{L}^{1}$-properties of these filters with application to
filter design.

\textbf{Notation}: The collection of signals whose $p^\mathrm{th}$ power is integrable
is denoted by $\mathcal{L}^{p}\left(  \mathbb{R}\right)  $ with the common
norm $\left\Vert \cdot\right\Vert _{p}$. For $p=\infty$ the norm is given by
the supremum norm. For further purposes let us also introduce the space of
bounded, continuous signals over $\mathbb{R}$, denoted by $C\left(
\mathbb{R}\right)  $ and endowed with the supremum norm.

\section{Preliminaries}

\subsection{Band-limited signals}

We start by investigating the standard band-limited setting: A signal is
called band-limited with bandwidth $B$ if the Fourier transform is supported on $\left[
-B,B\right]  $. The set of band-limited signals with bandwidth $B$ in
$\mathcal{L}^{p}\left(  \mathbb{R}\right)  $ form the Paley-Wiener space
$\mathcal{PW}_{B}^{p}$. The spaces $\mathcal{L}^{p}\left(  \mathbb{R}\right)
,C\left(  \mathbb{R}\right)  ,C_{T}\left(  \mathbb{R}\right)  $\ contain
generally signals of which the spectrum cannot be defined in the classical
sense so that it becomes distributional. Note that the inclusions
$\mathcal{PW}_{B}^{1}\subset\mathcal{PW}_{B}^{2}\subset\ldots\subset
\mathcal{PW}_{B}^{\infty}$ hold for Paley-Wiener spaces.

A signal $f\in$ $\mathcal{PW}_{B}^{p},1\leq p<\infty$, can be recovered by its
samples by applying the Shannon sampling series; denoting the Nyquist-rate
(critical) samples $t_{l}=\frac{\pi l}{B},l\in\mathbb{Z}$, we have
\begin{equation}
f\left(  t\right)  =\sum_{l=-\infty}^{\infty}f\left(  t_{l}\right)  \frac
{\sin\left[  B\left(  t-t_{l}\right)  \right]  }{B\left(  t-t_{l}\right)  }.
\label{eqn:shannon}%
\end{equation}
Unfortunately, the sampling series (\ref{eqn:shannon}) fails to converge in
general for $p=\infty$. In this case, for every signal $f\in\mathcal{PW}%
_{B}^{\infty}$ Sch\"{o}nhage's sampling series
\begin{align*}
f\left(  t\right)  &= f^{\prime}\left(  0\right)  \frac{\sin\left(
B\theta\right)  }{B}+f\left(  0\right)  \frac{\sin\left(  B\theta\right)
}{B\theta} \\
&\quad+t\sum_{l=-\infty,l\neq0}^{\infty}\frac{f\left(  t_{l}\right)  }%
{l}\frac{\sin\left[  B\left(  t-t_{l}\right)  \right]  }{B\left(
t-t_{l}\right)  }%
\end{align*}
can be applied converging uniformly on compact subsets of {$\mathbb{R}$}. Let
us define the following family of kernels:

\begin{definition}
A set $\mathcal{M}_{L_{\epsilon}}^{B}$ is called a \emph{reproducing kernel
set} if%
\[
\mathcal{M}_{L_{\epsilon}}^{B}:=\!\! \left\{  g\in L^{1}\left(  \mathbb{R}\right)
,\widehat{g}\left(  \omega\right)  =\!\!\left\{ \!
\begin{array}
[c]{cc}%
1 & \left\vert \omega\right\vert \leq B\\
\widehat{g}_{d}\left(  \omega\right)  & B\leq\left\vert \omega\right\vert \leq
L_{\epsilon}B\\
0 & \text{elsewhere}%
\end{array}
\right.  \right\}  ,
\]
where $\widehat{g}_{d}\left(  \omega\right)$ is a real function with
$0\leq\widehat{g}_{d}\left(  \omega\right)  \leq1$, $\widehat{g}_{d}\left(
B\right)  =1$, $\widehat{g}_{d}\left(  L_{\epsilon}B\right)  =0$.

The real number $L_{\epsilon}\geq1$ is called the \emph{bandwidth expansion
factor}.
\end{definition}

For some reasons that will become clear later on we assume $\widehat{g}%
_{d}\left(  \omega\right)  $ to be a non-increasing function. The set
$\mathcal{M}_{L_{\epsilon}}^{B}$ is in fact not empty; an example of a kernel
is given by the trapezoidal kernel
\begin{equation}
S_{L_{\epsilon}}\left(  t\right)  =\frac{2\sin\left(  \frac{\left(
L_{\epsilon}+1\right)  Bt}{2}\right)  \sin\left(  \frac{\left(  L_{\epsilon
}-1\right)  Bt}{2}\right)  }{\pi\left(  L_{\epsilon}-1\right)  Bt^{2}},
\label{eqn:trapezoidal_kernel}%
\end{equation}
of which the Fourier transform is depicted in Fig.~\ref{fig:fejer}a.

\begin{figure*}[t]
\centering
\includegraphics[width=0.7\textwidth, trim= 0.5cm 5.5cm 1cm 6cm, clip=true]{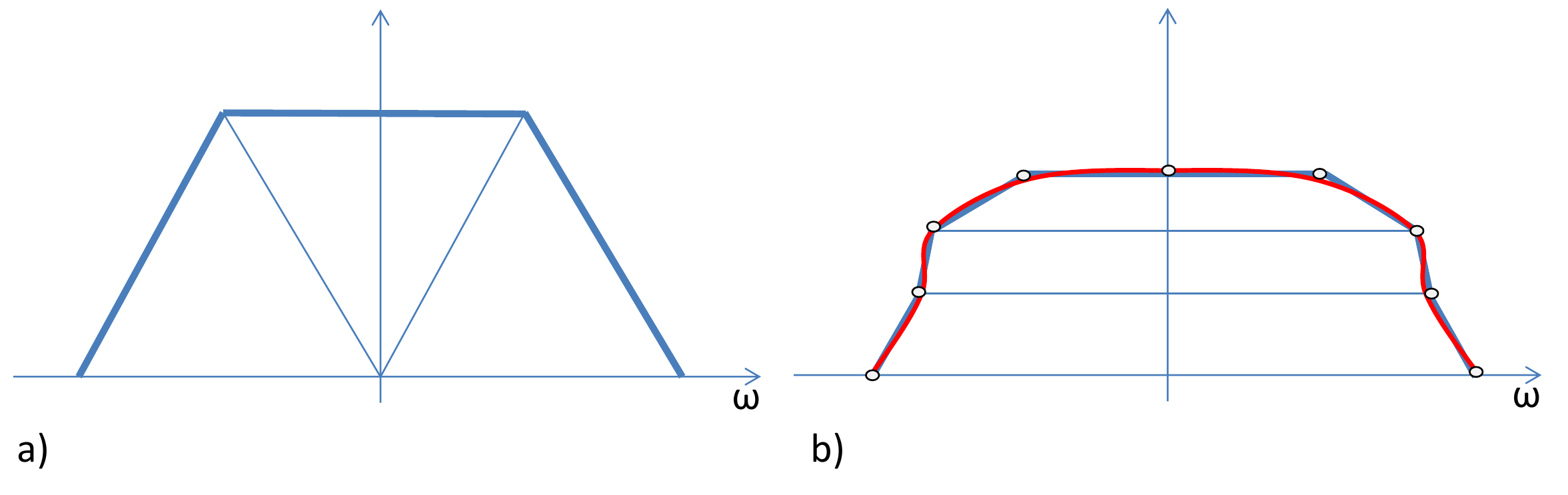}
\caption{a) The Fourier transform of the trapezoidal kernel ($L=L_{\epsilon}$), b) Approximation of Nyquist filters response by trapezoidal kernels.}
\label{fig:fejer}
\end{figure*}

The reproducing kernel set $\mathcal{M}_{L_{\epsilon}}^{B}$ can be
used together with oversampling (beyond Nyquist-rate sampling) for alternative
representations of Paley-Wiener spaces. Define $t_{l,L}:=\frac{\pi l}{LB}$
where the real number $L\geq1$ is the \emph{oversampling factor}. By a simple
application of the bounded convergence theorem and provided that $L\geq
L_{\epsilon}$ we have%
\[
f\left(  t\right)  =\frac{\pi}{LB}\sum_{l=-\infty}^{+\infty}f\left(
t_{l,L}\right)  g\left(  t-t_{l,L}\right)  ,\;f\in\mathcal{PW}_{B}^{\infty},
\]
i.e. the sampling of $f$ together with the kernel $g$ reproduces $f$. The
following theorem is a generalization.

\begin{theorem}
\label{theorem:poisson} For any $f\in\mathcal{PW}_{B}^{\infty}$ we have:
\begin{equation}
f\left(  t\right)  =\frac{\pi}{LB}\sum_{l=-\infty}^{+\infty}f\left(
t_{l,L}\right)  g\left(  t-t_{l,L}\right)  \label{eqn:shannon_l1}%
\end{equation}
where $g\in\mathcal{M}_{L_{\epsilon}}^{B},L\geq\frac{L_{\epsilon}+1}{2}$.
\end{theorem}

\begin{proof}
Since $g\in\mathcal{L}^{1}\left(  \mathbb{R}\right)  $ this can be shown using
the classical Poisson sum formula.
\end{proof}

\begin{remark}
Notably, if the kernel satisfies the Nyquist intersymbol interference (ISI)
criterion such that for some $(L_{\epsilon}+1)B/2\leq LB\leq L_{\epsilon}B$%
\begin{equation}
\sum_{k=-\infty}^{+\infty}\widehat{g}\left(  \omega-2LB\right)  =C_{N}%
\;\forall\omega, \label{eqn:NyquistISI}%
\end{equation}
then the kernel is called a \emph{Nyquist filter}, where $C_{N}>0$ is some
constant independent of $\omega$. $LB$ is called the Nyquist (angular)
frequency (half the sampling rate). Moreover, if such a filter is sampled
higher than as dictated by (\ref{eqn:NyquistISI}) then it is called
\emph{faster than Nyquist} (FTN) signaling.
\end{remark}

\subsection{Problem statement: Overshooting}

Overshooting is a classical problem in multicarrier communications and will be
referred to here as the \emph{peak value problem} (PVP). PVP investigates peak
regrowth between the samples of a (possibly oversampled) band-limited signal.
For a formal definition of the problem let $\left\Vert f\right\Vert _{\Delta
t,p}$ be $l_{p}$-sequence norm obtained from sampling $f$ with rate $2LB$.

\begin{problem}
\emph{Peak Value Problem:} Find a ``good\textquotedblright\ (i.e. tight) upper
bound on the constant
\[
C_{1}\left(  L\right)  :=\sup_{\left\Vert f\right\Vert _{t_{1,L},\infty}%
\leq1,f\in\mathcal{PW}_{B}^{\infty}}\left\Vert f\right\Vert _{\infty}.
\]

\end{problem}

It is interesting to note that $C_{1}\left(  L\right)  $ is independent of
$B$. We further note that $C_{1}\left(  L\right)  $ is not defined for $L=1$
(i.e. $C_{1}\left(  1\right)  =+\infty$) \cite{wunder_03_sig}. Eventually,
existence of a function $f$ such that $C_{1}\left(  L\right)  =\left\Vert
f\right\Vert _{\infty}$ is guaranteed \cite{boche_02_zamm}. We have from
\cite{wunder_03_sig}:

\begin{theorem}
\label{theorem:c1} Let $L>1$. We have:
\begin{equation}
C_{1}\left(  L\right)  =f_{L}^{\ast}\left(  t_{L}^{\ast}\right)  \leq\frac
{1}{\cos\left(  \frac{\pi}{2L}\right)  }. \label{eqn:c1_cos}%
\end{equation}
Moreover, if $L\in${$\mathbb{N}$}, then
\[
t_{L}^{\ast}=\frac{1}{2L},f_{L}^{\ast}\left(  t\right)  =\frac{\cos\left[
\pi\left(  t-\frac{1}{2L}\right)  \right]  }{\cos\left[  \frac{\pi}%
{2L}\right]  }%
\]
and%
\[
C_{1}\left(  L\right)  =\frac{1}{\cos\left(  \frac{\pi}{2L}\right)  }%
\]

\end{theorem}

Note that PVB is solved when restricting $\mathcal{PW}_{B}^{\infty}$ to%
\begin{align*}
\mathcal{T}_{N}:=\bigg\{  f\in\mathcal{PW}_{N}^{\infty};&\;f\left(
\theta\right)  =\frac{a_{0}}{2}+\sum_{k=1}^{N}a_{k}\cos\left(  k\theta\right) \\
&  +b_{k}\sin\left(  k\theta\right)  ,\quad a_{k},b_{k}\in\mathbb{R} \bigg\}  ,
\end{align*}
containing all degree $N$ trigonometric polynomials (with normalized angular
bandwidth $N$). Here, we define
\begin{equation}
C_{1}^{\mathcal{T}_{N}}\left(  N_{1}\right)  :=\max_{\left\Vert f\right\Vert
_{\frac{2\pi}{N_{1}},\infty}\leq1,f\in\mathcal{T}_{N}}\left\Vert f\right\Vert
_{\infty}, \label{eqn:c1_tn}%
\end{equation}
i.e. the optimizer set in the peak value problem is restricted to
$\mathcal{T}_{N}$. The extremal signals $f_{L}^{\ast}$ of (\ref{eqn:c1_tn})
are denoted as $\mathcal{T}$-$\left(  N,N_{1}\right)  $-extremal which can be
exactly calculated \cite{jetter_01_theo_approx}.

Altogether, PVP is quite well-understood. However, two problems occur in the
context of general multicarrier signals:

\begin{itemize}
\item[1)] The bounds in (\ref{eqn:c1_cos}), (\ref{eqn:c1_sqrt}) are
quite good for $L\geq2$; if $1\leq L\leq2$ the provided bounds are quite loose

\item The theorems require $L\geq L_{\epsilon}$ which is typically not the
case for Nyquist filters
\end{itemize}

The latter observations are a main motivation for the new approach derived in
the following sections.

\subsection{A related problem}

From (\ref{eqn:shannon_l1}), we obtain the following inequality
\[
\sup_{t\in\mathbb{R}}\left\vert f\left(  t\right)  \right\vert \leq\left\Vert
f\right\Vert _{t_{1,L},\infty}\cdot\sup_{t\in\left[  0,\frac{\pi}{LB}\right]
}\frac{\pi}{LB}\sum_{l=-\infty}^{+\infty}\left\vert g\left(  t-t_{l,L}\right)
\right\vert <\infty,
\]
i.e. every kernel $g\in\mathcal{M}_{L_{\epsilon}}^{B}$ defines a bounded,
linear operator $T_{g,L}:\mathcal{PW}_{B}^{\infty}\rightarrow\mathcal{PW}%
_{B}^{\infty}$. In communication context, the samples do not represent the
samples of a band-limited signal with respect to the bandwidth defined by $B$.
It is therefore reasonable to extend the definition region of the operator
$T_{g,L}$ to the space $C\left(  \mathbb{R}\right)  $, i.e.
\[
T_{g}^{L}:C\left(  \mathbb{R}\right)  \rightarrow\mathcal{PW}_{L_{\epsilon}%
B}^{\infty},f\hookrightarrow\frac{\pi}{LB}\sum_{l=-\infty}^{\infty}f\left(
t_{1,L}\right)  \,g\left(  t-t_{1,L}\right)  .
\]
The norm of this operator is given by
\[
\left\vert T_{g}^{L}\right\vert =\sup_{\left\Vert f\right\Vert _{\infty}%
\leq1,f\in C\left(  \mathbb{R}\right)  }\left\Vert T_{g}^{L}f\right\Vert
_{\infty}.
\]
The operator norm represents the enhancement of errors in the samples. This
leads us to the problem:\ 

\begin{problem}
\emph{Related problem:} Find a ``good\textquotedblright\ upper bound on the
operator norm
\[
C_{2}\left(  L,L_{\epsilon}\right)  :=\inf_{g\in\mathcal{M}_{L_{\epsilon}}%
^{B}}\left\vert T_{g}^{L}\right\vert .
\]
If $L=L_{\epsilon}$ then $C_{2}\left(  L,L_{\epsilon}\right)  :=C_{2}\left(
L\right)  $.
\end{problem}

Again, note that $C_{2}\left(  L,L_{\epsilon}\right)  $ is independent of $B$.
For the purposes of filter design it is also interesting which kernel actually
attains this bound \cite{wunder_03_sig}. These filters will be called
\emph{extremal filters} and their existence is established in the next theorem.

\begin{theorem}
\label{theorem:c2_exist} For any $L,L_{\epsilon}>1$ there is an extremal
signal $f^{\ast}$, a time instance $t^{\ast}$, and a kernel $g^{\ast}%
\in\mathcal{M}_{L_{\epsilon}}^{B}$ such that $C_{2}\left(  L,L_{\epsilon
}\right)  =(T_{g^{\ast}}^{L}f^{\ast})\left(  t^{\ast}\right)  $.
\end{theorem}

The proof is omitted and can be found in \cite{Wunder2003_PhD}.

The main connection of the related problem to PVB is clearly
\begin{equation}
C_{1}\left(  L\right)  \leq C_{2}\left(  L,L_{\epsilon}\right)  ,\;\text{
provided }L\geq\frac{L_{\epsilon}+1}{2}, \label{eqn:c1c2_rel}%
\end{equation}
i.e. $C_{1}\left(  L\right)  $ represents a lower bound on what can be
achieved for $C_{2}\left(  L,L_{\epsilon}\right)  $. Using this approach it
was shown in \cite{wunder_03_sig} that:

\begin{theorem}
\label{theorem:c2_sota} Suppose $L>1$. Then:
\begin{equation}
C_{2}\left(  L,L_{\epsilon}\right)  \leq\sqrt{\frac{L_{\epsilon}%
+1}{L_{\epsilon}-1}} \label{eqn:c2_sota}%
\end{equation}
Furthermore, for very small $L_{\epsilon}$:%
\[
C_{2}\left(  L,L_{\epsilon}\right)  =\frac{2}{\pi}\log\left(  \frac
{2L_{\epsilon}}{L_{\epsilon}-1}\right)  +O(1)
\]

\end{theorem}

Setting $L=L_{\epsilon}$ (as done in \cite{wunder_03_sig}), by virtue of 
(\ref{eqn:c1c2_rel}) we have
\[
C_{1}\left(  L\right)  \leq C_{2}\left(  L,L_{\epsilon}\right)  \leq
\sqrt{\frac{L+1}{L-1}}%
\]
and the bound is better for $L<2$ compared to the $1/\cos(\frac{\pi}{2L})$ law
but still quite loose. However, a careful analysis reveals that by Theorem~\ref{theorem:poisson}
 only $L\geq\frac{L_{\epsilon}+1}{2}$ is required so that
for the same $L$ the expansion factor $L_{\epsilon}$ can be pushed to
$L_{\epsilon}\leq2L-1$ and since the RHS of (\ref{eqn:c2_sota}) is monotone in
$L$ we obtain:
\begin{equation}
C_{1}\left(  L\right)  \leq\sqrt{\frac{2L}{2L-2}}=\sqrt{\frac{L}{L-1}}
\label{eqn:c1_sqrt}%
\end{equation}

\begin{remark}
In this setting, i.e. $L=2$, $L_{\epsilon}=2L-1=3$ the trapezoidal kernel is
the (optimal) extremal filter $g_{L_{\epsilon}}^{\ast}$ for the operator norm
$\left\vert T_{g,L}\right\vert $ since:
\[
C_{1}\left(  2\right)  =\frac{1}{\cos\left(  \frac{\pi}{2\cdot2}\right)
}=\sqrt{\frac{2}{2-1}}=\sqrt{2}%
\]

\end{remark}

We are now improving on this result.

\section{Main results}

In order to get an upper bound we make the following reasoning: Assume without
loss of generality $B=\pi$, let $K_{n}\in\mathcal{PW}_{\frac{\pi}{n}}^{2}%
\cap\mathcal{PW}_{\frac{\pi}{n}}^{1}$ for some natural $n\geq1$ be given as%
\[
K_{n}\left(  t\right)  :=\frac{2n\sin^{2}\left(  \frac{\pi t}{n2}\right)
}{\pi^{2}t^{2}}\geq0,
\]
hence some triangle kernel. The reason for the normalization of the bandwidth
to $\frac{\pi}{n}$ will become clear later on. At this point we could have
used $\pi$ (or any other value) as well. We need the following sub-sampling property.

\begin{lemma}
\label{lemma:fejer} Let $K_{n}\in\mathcal{PW}_{\frac{\pi}{n}}^{2}%
\cap\mathcal{PW}_{\frac{\pi}{n}}^{1}$ for some natural $n\geq1$, and
$K_{n}\left(  0\right)  =1,K_{n}\left(  \frac{\pi}{n}\right)  =0$. Then, we
have for some positive real $a\leq2$%
\[
\sum_{l=-\infty}^{+\infty}K_{n}\left(  t-lan\right)  =\frac{1}{an}.
\]

\end{lemma}

\begin{proof}
Let $\hat{K}_{n}$ be the spectrum of $K_{n}\in\mathcal{L}^{1}\left(
\mathbb{R}\right)  $. By the Poisson sum formula we have:%
\[
\sum_{m=-\infty}^{+\infty}K_{n}\left(  t-man\right)  =\frac{1}{an}%
\sum_{k=-\infty}^{+\infty}\hat{K}_{n}\left(  \omega_{k}\right)  e^{j\omega
_{k}t}%
\]
with $\omega_{k}=\frac{2\pi k}{an}$. By assumption $\hat{K}_{n}\left(
\omega_{k}\right)  =0$ for $k\geq1$, which gives the final result.
\end{proof}

The following theorem is an upper bound on $C_{2}\left(  L\right)  $
generalizing \cite[Theorem 4.10]{litsyn_07}.

\begin{theorem}
\label{theorem:c2} Let $L_{\epsilon}=\frac{n+1}{n},L=\frac{n+m}{n}%
,n\in\mathbb{N},m\in\mathbb{N}\cup\{\frac{1}{2}\}$. Then:
\begin{align}
C_{2}\left(  L,L_{\epsilon}\right) & \leq\max_{-\frac{n}{2\left(  n+1\right)} \leq t\leq\frac{n}{2\left(  n+1\right)  }}\frac{1}{2\left(  n+m\right)  } \nonumber \\
&\quad\quad\sum_{l=0}^{2\left(  n+m\right)  -1}\left\vert \sum_{k=-n}^{n}e^{jk\left(
\frac{\pi t}{n}-\frac{l\pi}{n+m}\right)  }\right\vert 
\label{eqn:c2_new}%
\end{align}
Moreover:%
\[
C_{2}\left(  L\right)  \geq C_{1}^{\mathcal{T}_{N}}\left(  N_{1}\right)
\]
Here, $N,N_{1}$ are taken from any uniform sampling of (frequency) support
$\left[  0,\frac{n+1}{n}\pi\right]  $ of some $g\in\mathcal{M}_{\frac{n+1}{n}%
}^{\pi}$ where all frequencies that fall in the interval where the response
equals unity gives the highest in-band frequency $N$ and all that fall
out-of-band give $N_{2}-N$.
\end{theorem}

\begin{proof}
For the purpose of practical applicability let us assume a more general
setting. We assume that the kernel is in $g\in\mathcal{M}_{L_{\epsilon}}^{\pi
}$ where $L_{\epsilon}:=\frac{n+1}{n},n\in\mathbb{N},$ and let the
oversampling factor to be of the form%
\[
L=\frac{n+m}{n}%
\]
for some $n\in\mathbb{N},m\in\mathbb{N}\cup\{\frac{1}{2}\}$. Then, the
trapezoidal kernel can be represented as%
\[
g_{L_{\epsilon}}\left(  t\right)  =\sum_{k=-n}^{n}K_{n}\left(  t\right)
e^{j\frac{k\pi t}{n}}%
\]
where $K_{n}$ is a kernel given by%
\[
K_{n}\left(  t\right)  =\frac{2n\sin^{2}\left(  \frac{\pi t}{n2}\right)  }%
{\pi^{2}t^{2}}\geq0,
\]
the so-called triangle kernel, which can be seen as a special case of the
trapezoidal kernel \cite{Wunder2015_SPAWC}. This is illustrated in Fig.~\ref{fig:fejer}a. Clearly, $K_{n}%
\in\mathcal{PW}_{\frac{\pi}{n}}^{2}\cap\mathcal{PW}_{\frac{\pi}{n}}^{1}$.
Hence, we obtain
\begin{align*}
C_{2}\left(  L\right)   &  =\sup_{t\in \Omega}\frac{1}{L}\sum_{l=-\infty}^{\infty
}\left\vert g_{L_{\epsilon}}\left(  t-l\frac{1}{L}\right)  \right\vert \\
&  =\max_{t\in\Omega}\frac{1}{2\left(  n+m\right)  }\sum_{l=0}^{2\left(  n+m\right)
-1}\left\vert \sum_{k=-n}^{n}e^{jk\left(  \frac{\pi t}{n}-\frac{l\pi}%
{n+m}\right)  }\right\vert
\end{align*}
where $\Omega=[t: -\frac{n}{2\left(  n+1\right)  }\leq
t\leq\frac{n}{2\left(  n+1\right)  }]$. 
Since the term $|\sum_{k=-n}^{n}e^{jk\left(  \frac{\pi t}{n}-\frac{l\pi}%
{n+m}\right)  }|$ is periodic with $2\left(  n+m\right)  $ (we sample it at
$\frac{l\pi}{n+m}$), $K\left(  t\right)  \geq0$, and finally%
\begin{align*}
\frac{n}{n+k}\sum_{l=-\infty}^{\infty} &K_{n}\left(  t-\frac{nl\cdot2\left(
n+m\right)  }{n+m}\right)   \\
&  =\frac{1}{2\left(  n+m\right)  }%
\end{align*}
due to Lemma~\ref{lemma:fejer}.

The construction of the lower bound is omitted due to lack of space.
\end{proof}

For $L_{\epsilon}\neq L$ we obtain the following corollary:

\begin{corollary}
\label{corollary} Let $L=\frac{n+1}{n},\;n$ even. Then:%
\[
C_{1}\left(  L\right)  \leq\max_{-\frac{n}{2\left(  n+1\right)  }\leq
t\leq\frac{n}{2\left(  n+1\right)  }}\frac{1}{n+1}\sum_{l=0}^{n}\left\vert
\sum_{k=-n/2}^{n/2}e^{jk\left(  \frac{2\pi t}{n}-\frac{l2\pi}{n+m}\right)
}\right\vert
\]

\end{corollary}

\begin{proof}
Setting $m=1/2$ yields $L=(2n+1)/2n$. Replace $n^{\prime}=2n$ so that
$L_{\epsilon}$ can be pushed towards $L_{\epsilon}:=(n^{\prime}+2)/n^{\prime}$
yields the result provided $n^{\prime}$ is even and since $L\geq(L_{\epsilon
}+1)/2$ for any $n$.
\end{proof}

Numerical computations of the operator norm for different trigonometric
polynomials were carried out and shown in Fig.~\ref{fig:eval_norm} along with
the bounds in (\ref{eqn:c1_cos}), (\ref{eqn:c1_sqrt}),
(\ref{eqn:c2_sota}) and the new bound in (\ref{eqn:c2_new}) where, for
the sake of simplicity in both figures, the curves are depicted over
$\mathbb{R}$. It is observed that there is strong improvement of the new over
the existing bounds.

\begin{figure}[ptb]
\centering
\includegraphics[width=\columnwidth]{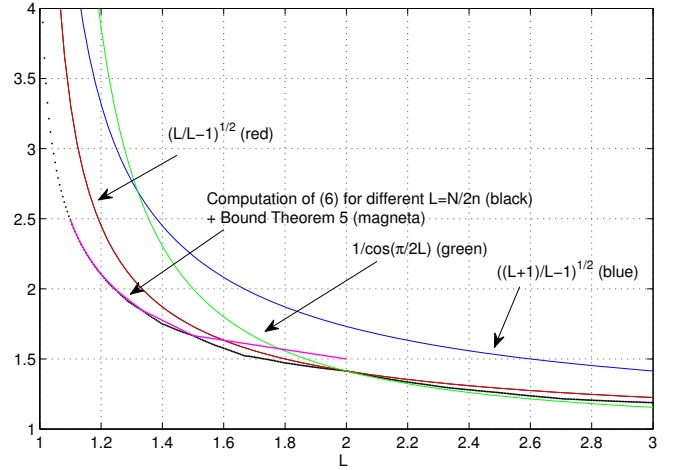}
\caption{Bounds on $C_{2}(L),C_{2}(L,L_{\epsilon})$}%
\label{fig:eval_norm}
\end{figure}

\section{Extensions: Overshooting and $\mathcal{L}^{1}$-Norms of Nyquist
filters}

\subsection{Overshooting of Nyquist filters}

Trapezoidal filters have been used in \cite{Wunder2015_SPAWC} as Nyquist
filters. Obviously, to bound the overshooting for Nyquist filters we cannot
use the classical PVP results since $L\leq L_{\epsilon}$ when the Nyquist
criterion is enforced, i.e. the transmit signal is actually undersampled with
respect to the transmit signal's bandwidth. On the other hand, the new
approach can be used, specifically (\ref{eqn:c2_new}) and
(\ref{eqn:c2_sota}). More specifically, if $L_{\epsilon}=\frac{n+1}{n}%
,L=\frac{n+m}{n},n\in\mathbb{N},m\in\mathbb{N}\cup\{\frac{1}{2}\}$,
overshooting between the samples (i.e. data sequence) is upperbounded by
(without loss of generality still $B=\pi$)%
\begin{align}
\min&\Bigg\{  \sqrt{2n+1}\ ,\max_{-\frac{n}{2\left(  n+1\right)  }\leq t\leq
\frac{n}{2\left(  n+1\right)  }} \nonumber\\
&\quad\quad\quad \frac{1}{2\left(  n+m\right)  }\sum
_{l=0}^{2\left(  n+m\right)  -1}\left\vert \sum_{k=-n}^{n}e^{jk\left(
\frac{\pi t}{n}-\frac{l\pi}{n+m}\right)  }\right\vert \Bigg\}
\label{eqn:NyquistOvershoot}%
\end{align}
where we tacitly assumed that the data is bounded by unity. It is also
possible to construct more general filters which are ``layerwise" composed of
trapezoidal filters (or approximations of them), see the Fig.~
\ref{fig:fejer}b. Then we can bound the overshooting as follows:%

\begin{align*}
C_{2}\left(  L\right)   &  =\sup_{-\frac{1}{2L}\leq t\leq\frac{1}{2L}}\frac
{1}{L}\sum_{l=-\infty}^{\infty}\left\vert g_{L_{\epsilon}}\left(  t-\frac
{l}{L}\right)  \right\vert \\
&  \leq\sum_{k=0}^{n-1}(  X\left(  \omega_{k}\right)  -X\left(
\omega_{k+1}\right)  )  \\
&\quad\quad\quad\quad \underbrace{\sup_{-\frac{1}{2L}\leq t\leq
\frac{1}{2L}}\frac{1}{L}\sum_{l=-\infty}^{\infty}\left\vert S_{L_{\epsilon
}^{k}}\left(  t-\frac{l}{L}\right)  \right\vert }_{\text{bound by
(\ref{eqn:NyquistOvershoot})}}%
\end{align*}
Here, recall that $S_{L_{\epsilon}^{k}}$ is the trapezoidal kernel family,
$X\left(  \omega_{k}\right)  $ are the frequency sampling points in the decay
region, such that $X\left(  \omega_{k}\right)  -X\left(  \omega_{k+1}\right)
\geq0$ (setting $X(\omega_{k+1})=0$), and $L_{\epsilon}^{k}$ are the extension
factors, respectively. Hence, we conclude that for such general filters the
overshoot is just an average of the individual trapezoidal layers with
extension factors $L_{\epsilon}^{k},k=0,...,n-1$.

%
\subsection{$\mathcal{L}^{1}$-Norms of Nyquist filters}

The Nyquist criterion is a fundamental property ensures that samples are not
interfering with each other. In a practical system nevertheless the samples
actually do so as a consequence of the impairments of the communication
channel such as channel induced ISI, time/frequency offsets etc. Notably, it
is deliberately induced in FTN signalling. The impact of these effects is
measured by the opening in the eye diagram of the overlayed signal, and is
directly related to the $\mathcal{L}^{1}$-norm (tails) of the used Nyquist or
(FTN) filter.

Let us comment on the relation to the $\mathcal{L}^{1}$-norm of trapezoidal
filters and related families which bounds the ISI. Note first that, for the
triangle kernel, this norm is actually independent and unity for all $B$
since:%
\[
\lVert g\rVert_{1}=\frac{B}{2\pi}\int\frac{B}{2\pi}\text{sinc}^{2}\left(
\frac{Bt}{2}\right)  dt=1
\]
Here, we defined sinc$\left(  At\right)  =\sin(At)/At,A>0$. It is easy to
prove that no filter with $g(0)=1$ can fall below this value. Hence, we can
argue that for any filter $\lVert g\rVert_{1}>1$ (after proper normalization),
and that the $\mathcal{L}^{1}$-norm measures the (inverse) distance to the
sinc kernel for which clearly $\lVert g\rVert_{1}=\infty$ holds. Suppose we
want to ``shift" the triangle filter kernel to the optimal sinc kernel while
not loosing the favorable properties of the triangle kernel. This can achieved
by the trapezoidal kernel family which has close to optimal $\mathcal{L}^{1}%
$-norm behaviour as follows:

\begin{theorem}
We have:%
\[
\inf_{g\in\mathcal{M}_{L_{\epsilon}}^{B}}\left\Vert g\right\Vert _{1}\leq
C_{2}\left(  L\right)
\]

\end{theorem}

\begin{proof}
It is easy to see that we have for any $g\in\mathcal{PW}_{L_{\epsilon}B}^{1}$:%
\begin{align*}
\left\Vert g\right\Vert _{1}    =\int\limits_{\mathbb{R}}\left\vert g\left(
t\right)  \right\vert dt = \sum_{l=-\infty}^{+\infty}\int\limits_{0}^{t_{1,L}}\left\vert g\left(
t-t_{l,L}\right)  \right\vert dt
\end{align*}
By the bounded convergence theorem:%
\begin{align*}
  \sum_{l=-\infty}^{+\infty}\int\limits_{0}^{t_{1,L}}\left\vert g\left(
t-t_{l,L}\right)  \right\vert dt
 & =\int\limits_{0}^{t_{1,L}}\sum_{l=-\infty}^{+\infty}\left\vert g\left(
t-t_{l,L}\right)  \right\vert dt\\
&  \leq\max_{t\in\left[  0,\frac{\pi}{LB}\right]  }\frac{\pi}{LB}%
\sum_{l=-\infty}^{+\infty}\left\vert g\left(  t-t_{l,L}\right)  \right\vert
\end{align*}
Taking the limes inferior on both sides yields the result.
\end{proof}

We have the following observations which are quite convenient for filter
design in 5G :

\begin{itemize}
\item The $\mathcal{L}^{1}$-norm is almost independent of the actual filter
bandwidth so that tail properties can be tightly controlled when scaling.

\item Trapezoidal kernels seemingly yield a good comprise between properties
of both the extreme sinc and triangle kernels.

\item $\mathcal{L}^{1}$-norm are seemingly quite good but optimality is yet to
be proven.
\end{itemize}

\section{Conclusions}

In this paper we provided new bounds on the overshooting between samples of
bandlimited signals for small oversampling factors improving on former
results. Moreover, we discussed some extension to overshooting of Nyquist
filters and related ISI bounds.

\section*{Acknowledgements}

This work was supported by German Research Foundation/Deutsche Forschungsgemeinschaft (DFG) under grant WU 598/3-1.

\bibliographystyle{IEEEtran}
\bibliography{overshoot_25-09-2015_v7}

\end{document}